%% file: main.tex
\newif  \iflong 
\longtrue 

\iflong
\documentclass[runningheads]{llncs}
\usepackage{authblk}

\else

\documentclass[conference]{IEEEtran}
\newtheorem{definition}{Definition}
\newtheorem{theorem}{Theorem}
\newenvironment{proof}{{\sl Proof}:}{\hspace*{\fill}$\Box$}
\fi

\usepackage{cite}
\usepackage{amsmath,amssymb,amsfonts}
\usepackage{algorithmic}
\usepackage{graphicx}
\usepackage{textcomp}
\usepackage{xcolor}
\def\BibTeX{{\rm B\kern-.05em{\sc i\kern-.025em b}\kern-.08em
    T\kern-.1667em\lower.7ex\hbox{E}\kern-.125emX}}
\usepackage{multirow}
\usepackage{hyperref}
\usepackage[utf8]{inputenc}
\usepackage{macros}

\begin{document}

\title{Quasi-Cyclic Stern Proof of Knowledge}
\date{}

\iflong
\author{
Loïc Bidoux \inst{1} \and
Philippe Gaborit \inst{2} \and
Mukul Kulkarni \inst{1} \and
Nicolas Sendrier \inst{3}
}

\institute{
Technology Innovation Institute \and
University of Limoges \and 
INRIA \\
}

\fi

\maketitle

\begin{abstract}
  The ongoing NIST standardization process has shown that Proof of Knowledge (PoK) based signatures have become an important type of possible post-quantum signatures.
  Regarding code-based cryptography, the main original approach for PoK based signatures is the Stern protocol which allows to prove the knowledge of a small weight vector solving a given instance of the Syndrome Decoding ($\SD$) problem over $\Ft$.
  It features a soundness error equal to $2/3$.
  This protocol was improved a few years later by V\'eron who proposed a variation of the scheme based on the General Syndrome Decoding ($\GSD$) problem which leads to better results in term of communication.
  A few years later, the AGS protocol introduced a variation of the V\'eron protocol based on Quasi-Cyclic (QC) matrices.
  The AGS protocol permits to obtain an asymptotic soundness error of $1/2$ and an improvement in term of communications.

  In the present paper, we introduce the Quasi-Cyclic Stern PoK which constitutes an adaptation of the AGS scheme in a $\SD$ context, as well as several new optimizations for code-based PoK.
  Our main optimization on the size of the signature can't be applied to $\GSD$ based protocols such as AGS and therefore motivated the design of our new protocol.
  In addition, we also provide a special soundness proof that is compatible with the use of the Fiat-Shamir transform for 5-round protocols.
  This approach is valid for our protocol but also for the AGS protocol which was lacking such a proof.
  We compare our results with existing signatures including the recent code-based signatures based on PoK leveraging the MPC in the head paradigm. 
  In practice, our new protocol is as fast as AGS while reducing its associated signature length by $20 \%$.
  As a consequence, it constitutes an interesting trade-off between signature length and execution time for the design of a code-based signature relying only on the difficulty of the $\SD$ problem.
\end{abstract}

\section{Introduction}

Since its introduction in 1978 by McEliece \cite{McEliece78}, code-based cryptography has been one of the main alternative to classical cryptography.
This is illustrated by the ongoing NIST Post-Quantum Cryptography standardization process whose round~3 features three code-based Key Encapsulation Mechanisms (KEM) \cite{ClassicMcEliece, BIKE, HQC}.
Additional KEM \cite{LEDA,ROLLO,RQC} were also considered during the round~2 of the competition. 
Although there exists satisfactory code-based KEM, designing signatures from coding theory has historically been challenging.
Two approaches have been used in the literature namely signatures from the hash-and-sign paradigm and signatures based on identification.
In the first category, a construction was proposed in 2001 \cite{CFS01} although it is rather inefficient.
The recent Wave construction \cite{wave19} provides an efficient solution
following the same paradigm and features small signature sizes.
In the second category, two constructions have been proposed in the past few years namely LESS~\cite{less} and Durandal~\cite{Durandal}.
Hereafter, we focus on signatures constructed from the Fiat-Shamir paradigm \cite{FS, DPJS96, ExtendedFS} along with Zero-Knowledge Proofs of Knowledge (ZK PoK) for the Syndrome Decoding ($\SD$) problem.

The first PoK for the $\SD$ problem was introduced by Stern in 1993 \cite{Stern93}.
In 1997, V\'eron showed that using the general decoding problem ($\GSD$), one can design a protocol that is more efficient than the initial Stern proposal \cite{Veron97}.
The $\SD$ and $\GSD$ problems are equivalent and differ only in the way used to represent the underlying code namely using a parity-check matrix in the former and using a generator matrix in the latter.
Both protocols require 3 rounds to be executed and feature a soundness error equal to 2/3.
In 2011, the CVE~\cite{CVE11} and AGS~\cite{AGS11} PoK respectively improved the Stern and V\'eron protocols by lowering their soundness error to 1/2 (asymptotically close to 1/2 for AGS) using 5 rounds of execution.
The CVE protocol is based on the $\SD$ problem over $\Fq$ while the AGS protocol relies on the $\QCGSD$ problem namely the $\GSD$ problem instantiated with a Quasi-Cyclic~(QC) matrix.
An issue with respect to the zero-knowledge property of $\GSD$ based  protocols (Véron and AGS) has been identified in \cite{ZKissue} and have been fixed in \cite{ISIT21}.
Some of these protocols have also been adapted in the rank metric setting, see \cite{Chen95, RankStern11, RankVeron19} for instance.
Recently, several proposals have used the MPC in the head paradigm in order to achieve a negligible soundness error of $1/N$ for some parameter~$N$.
The GPS \cite{GPS21} construction achieves such a small soundness error by relying on the SD problem over $\Fq$ while the FJR \cite{FJR21} and BGKM \cite{BGKM22} proposals rely on the SD problem over $\Ft$.
However, these constructions induce a performance overhead with respect to previous approaches.

Thanks to these new results, the research problem associated to these protocols has shifted from minimizing the signature size to finding the best trade-off between expected performances and signature size.
Amongst existing constructions, AGS features the smallest expected performance cost while FJR is the best approach to get small signature sizes.
In this paper, we propose a new PoK that has the same cost as AGS while featuring a signature size that is $20 \%$ smaller.
As such, our new protocol provides a new interesting trade-off for the design of signatures based on PoK for the $\SD$ problem.

\vspace{0.5\baselineskip}
\noindent \textbf{Contributions.} We introduce the Quasi-Cyclic Stern protocol which is a new PoK for SD problem as well as several new optimizations for code-based PoK.
Our main optimization on the size of the signature can't be applied to $\GSD$ based protocols such as AGS which motivates the design of our new protocol.
In addition, we also provide a special soundness proof that is compatible with the use of the Fiat-Shamir transform for 5-round protocols which was lacking in the AGS protocol.
In practice, our new protocol is as fast as AGS while reducing its communication length by $20 \%$ therefore providing an interesting trade-off for the design of a code-based signature relying only on the difficulty of the SD problem.

\iflong
\vspace{0.5\baselineskip}
\noindent \textbf{Paper organization.} 
We introduce some preliminaries on code-based cryptography and PoK in Section \ref{sec:preliminaries}.
Then, we describe the Quasi-Cyclic Stern protocol as well as new optimizations for code-based PoK in Section \ref{sec:protocol}.
We give some parameters and depict resulting key sizes and signature sizes in Section \ref{sec:parameters}.
We explain how to generalize our construction to other metrics in Section \ref{sec:generalization}.
\fi

\section{Preliminaries} \label{sec:preliminaries}

\iflong
We start by presenting some definitions related to code-based cryptography in Section \ref{sec:preliminaries:cbc}.
Then, we introduce zero-knowledge proofs of knowledge and explain how one can transform them into signatures in Section \ref{sec:preliminaries:pok}.
Finally, we describe the Stern protocol in Section \ref{sec:preliminaries:stern}.

\vspace{0.5\baselineskip}
\noindent \textbf{Notations.} Hereafter, we represent vectors (respectively matrices) using bold lower-case (respectively upper-case) letters.
\fi
%
We denote by $\hw{\bm{x}}$ the Hamming weight of $\bm{x}$ and by $\hperm$ the symmetric group of all permutations of $n$ elements. 
If $X$ is a finite set, then $x \sampler X$ denotes that $x$ is sampled uniformly at random from $X$ and $x \samples{\psi} X$ denotes that $x$ is sampled uniformly at random from $X$ using the seed $\psi$.

\subsection{Coding Theory and Cryptography} \label{sec:preliminaries:cbc}

\iflong
We start by defining linear codes and quasi-cyclic codes.
Next, we describe the syndrome decoding ($\SD$) and quasi-cyclic syndrome decoding ($\QCSD$) problems which are two difficult problems commonly used in code-based cryptography.
The $\SD$ problem has been proven NP-complete in \cite{BMVT78}.

\begin{definition}[Binary Linear Code]
Let $n$ and $k$ be positive integers such that $k < n$. A binary linear $\mathcal{C}$ code (denoted $[n,k]$) is a $k$-dimensional subspace of $\Ftn$.
  $\mathcal{C}$ can be represented in two equivalent ways: by a generator matrix $\bm{G} \in \Ft^{\ktn}$ such that $\mathcal{C} = \{\bm{mG} ~|~ \bm{m} \in \Ftk\}$ or by a parity-check matrix $\bm{H} \in \Ft^{\nmktn}$ such that $\mathcal{C} = \{ \bm{x} \in \Ftn ~|~ \bm{H}\bm{x}^{\top} = 0\}$.
\end{definition}

\begin{definition}[Systematic Quasi-Cyclic Code]
  A systematic quasi-cyclic code of index $\ell$ and rate $1/\ell$ is a
  $[n=\ell k, k]$ code with an $(\ell-1)k \times \ell k=(n-k)\times n$ parity check matrix of the form:
\begin{equation*}
  \mathbf{H}=
\begin{bmatrix}
  \mathbf{I}_k & 0   &  \cdots &0& \mathbf{A}_0\\
  0   & \mathbf{I}_k &        & & \mathbf{A}_1\\
      &     &  \ddots & & \vdots\\
  0   &     &   \cdots & \mathbf{I}_k     & \mathbf{A}_{\ell-2}
\end{bmatrix}
\end{equation*}
where $\mathbf{A}_0,\ldots ,\mathbf{A}_{\ell-2}$ are circulant $k\times k$ matrices.
\end{definition}

\begin{definition}[$\SD$ problem]
  Given positive integers $n$, $k$, $w$, a random parity-check matrix $\bm{H} \sampler \Ft^{\nmktn}$, and a syndrome $\bm{y} \in \Ft^{\nmk}$ the syndrome decoding problem $\SD(n,k,w)$ asks to find $\bm{x} \in \Ftn$, such that $\bm{Hx}^\top = \bm{y}^\top$ and $\hw{\bm{x}} = w$.
\end{definition}

\begin{definition}[$\QCSD$ problem]
  Given positive integers $n, k, w$, with $n=\ell k$ for some $\ell$, a random parity-check matrix of a quasi-cyclic code $\bm{H} \sampler \mathcal{QC}(\Ft^{(n - k) \times n})$, and a syndrome $\bm{y} \in \Ft^{(n-k)}$, the syndrome decoding problem $\QCSD(n,k,w)$ asks to find $\bm{x} \in \Ft^{n}$, such that $\bm{Hx}^\top = \bm{y}^\top$ and $\hw{\bm{x}} = w$.
\end{definition}
The quasi-cyclic codes we will consider hereafter will
always have index $\ell=2$.

\fi

\iflong
\subsection{Signatures from Zero-Knowledge Proof of Knowledge} \label{sec:preliminaries:pok}

We start by introducing commitment schemes as they are a building block used to construct proofs of knowledge.
We require such schemes to be computationally hiding and computationally biding. 
The former property ensures that efficient adversaries can't distinguish between two commitments generated from different messages. 
The latter property ensures that efficient adversaries can't change their committed messages after the commit step.

\fi

\iflong
\begin{definition}[Commitment Scheme]
  A commitment scheme is a tuple of algorithms $\COM = (\mathsf{Commit}, \mathsf{Open})$ such that $\mathsf{Commit}(m)$ returns a commitment $c$ for the message $m$ and $\mathsf{Open}(c, m)$ returns either $1$ ($\mathsf{accept}$) or $0$ ($\mathsf{reject}$). Note that, in general commitment schemes
  also take some randomness $r$ as input to the $\mathsf{Commit}$ and
  $\mathsf{Open}$ algorithms, we assume that this is implicit in the 
  subsequent definitions and discussion.
A commitment scheme is said to be correct if:
 \begin{equation*}
 \prb
  \left[ b = 1 \ \middle\vert \begin{array}{l}
    c \leftarrow \commit{m}, ~ b \leftarrow \open{c,m}
  \end{array}\right]
  =1.
  \end{equation*}
\end{definition}

\begin{definition}[Computationally Hiding]
Let $(m_0, m_1)$ be a pair of messages, the advantage of an adversary $\adv$ against the commitment hiding experiment be defined as:
\begin{equation*}
  \mathsf{Adv}^{\mathsf{hiding}}_{\adv}(1^\lambda) = 
  \Bigg| \, \prb \left[ \begin{array}{l}
    b = b' \\
  \end{array} \ \middle\vert~ \\ \begin{array}{l}
    b \sampler \{0, 1\}, ~ c \samplen \commit{m_b} \\
    b' \samplen \adv.\mathsf{guess}(c)
  \end{array}\right]
  - \frac{1}{2} \, \Bigg|.
\end{equation*}
A commitment scheme $\COM$ is computationally hiding if for all probabilistic polynomial time adversaries $\adv$ and every pair of messages $(m_0, m_1)$, $\mathsf{Adv}^{\mathsf{hiding}}_{\adv}(1^\lambda)$ is negligible in $\lambda$.
\end{definition}

\begin{definition}[Computationally Binding]
Let the advantage of an adversary $\adv$ against the commitment binding experiment be defined as:
\begin{equation*}
  \mathsf{Adv}^{\mathsf{binding}}_{\adv}(1^\lambda) = \prb
  \left[ \begin{array}{l}
    m_0 \neq m_1 \\
    1 \samplen \open{c, m_0} \\ 
    1 \samplen \open{c, m_1} \\ 
  \end{array} \ \middle\vert~ \\ \begin{array}{l}
    (c, m_0, m_1) \samplen \adv.\mathsf{choose}(1^{\lambda})
  \end{array}\right].
\end{equation*}
A commitment scheme $\COM$ is computationally binding if for all probabilistic polynomial time adversaries $\adv$, $\mathsf{Adv}^{\mathsf{binding}}_{\adv}(1^\lambda)$ is negligible in $\lambda$.
\end{definition}

We now describe proofs of knowledge along with their soundness and zero-knowledge properties. 
Informally, a zero-knowledge proof of knowledge allows a prover $\prover$ to prove the knowledge of a secret to a verifier $\verifier$ without revealing anything about it.
\fi
%
An interactive proof system is a protocol between two parties
($\prover$ and $\verifier$) used to establish the validity of some statement
$x$ by proving the existence of a witness $w$ such that $R(x,w) = 1$
for some public relation $R$ (\textit{i.e.} $x \in L$ for some language
$L \in \mathsf{NP}$). A PoK system
additionally proves that $\prover$ actually \emph{knows} a valid
witness $w$ (as opposed to its existence earlier). A PoK system is  
(1) \emph{complete} if proof corresponding to a valid statement ($x \in L)$
is always accepted by the honest verifier, 
(2) \emph{sound} if malicious prover cannot prove a false statement, 
(3) \emph{special sound} if repeated interaction with malicious prover (with a fixed  false statement) allows efficient recovery of a valid witness, and
(4) \emph{zero-knowledge} if the verifier does not learn any information
about the witness $w$ (beyond its existence) after interacting with the prover.

\iflong


\begin{definition}[Proof of Knowledge]
  A $(2n+1)$-rounds proof of knowledge $\mathsf{PoK} = (\setup, \allowbreak \keygen, \allowbreak \prover_1, \allowbreak \cdots, \allowbreak \prover_{n+1}, \allowbreak \verifier_1, \allowbreak \cdots, \allowbreak \verifier_{n+1})$ is an interactive protocol between a prover $\prover$ and a verifier $\verifier$.
We denote by $\big \langle \prover(\param, \allowbreak \sk, \allowbreak \pk), \allowbreak \verifier(\param, \allowbreak \pk) \big \rangle$ the transcript of a proof of knowledge between a prover $\prover$ and a verifier $\verifier$. 

\vspace{0.5\baselineskip}
\noindent A proof of knowledge is correct if for all $\param$ returned by $\setup(1^\lambda)$ and all $(\sk, \pk)$ returned by $\keygen(\param)$:
\begin{equation*}
  \prb
  \left[ \begin{array}{l}
    \accept \leftarrow \big \langle \prover(\param, \sk, \pk), \\ \verifier(\param, \pk) \big \rangle
  \end{array} \ \middle\vert \\ \begin{array}{l}
    \param \leftarrow \setup(1^\lambda) \\
    (\sk, \pk) \leftarrow \keygen(\param)
  \end{array}\right]
  = 1.
\end{equation*}
\end{definition}

\begin{definition}[Soundness] A proof of knowledge is sound with soundness error $\epsilon$ if for all $\param$ returned by $\setup(1^\lambda)$, all $(\sk, \pk)$ returned by $\keygen(\param)$ and all probabilistic polynomial time $\ppt$ malicious prover $\tilde{\prover}$:
\begin{equation*}
 \prb
  \left[ \begin{array}{l}
    \accept \leftarrow \big \langle \tilde{\prover}(\param, \pk), \\ \verifier(\param, \pk) \big \rangle
  \end{array} \ \middle\vert \begin{array}{l}
    \param \leftarrow \setup(1^\lambda) \\
    (\sk, \pk) \leftarrow \keygen(\param) 
  \end{array}\right]
  \leq \epsilon + \negl(\lambda).
\end{equation*}
\end{definition}

\begin{definition}[Honest-Verifier Zero-Knowledge]
  A proof of knowledge satisfies the honest-verifier zero-knowledge (HZVK)  property if there exists a $\ppt$ simulator $\simulator$ that given as input $(\param, \pk)$, outputs a transcript $\big \langle \simulator(\allowbreak \param, \allowbreak \pk), \allowbreak \verifier(\param, \allowbreak \pk) \big \rangle$ that is computationally indistinguishable from the probability distribution of transcripts of honest executions between a prover $\prover$ with input $(\param, \sk, \pk)$ and verifier $\verifier$ with input $(\param, \pk)$.
\end{definition}

Here, we will assume that the $(2n + 1)$ round PoK protocol starts with message $a$ from the prover $\prover$ and it is followed by
subsequent challenge and response messages. The challenge message $c_i$ is the challenge sent by the verifier $\verifier$ in $i$-th round
and $z_i$ is the response to $c_i$ sent by $\prover$.
We also call the tuple $(a, c_1, z_1, \ldots, c_n, z_n)$ as \emph{transcript} of the protocol. The transcript is called
\emph{accepting} if $\verifier$ accepts the proof. \footnote{We also assume that the protocols discussed in this work are public-coin, 
that is the randomness used by $\verifier$ to compute challenges is public.}

\begin{definition} [Tree of Transcripts~\cite{EPRINT:AttFehKlo21}] \label{def:tree-of-transcripts}
  Let $k_1, k_2, \ldots, k_n \in \mathbb{N}$. A $(k_1, k_2, \ldots, k_n)$-tree of transcripts for a $(2n+1)$-move public coin protocol
  $\mathsf{PoK}$ is a set of $K := \prod_{i = 1}^{n} k_i$ transcripts arranged in the following tree structure. The nodes in the tree
  represent the prover's messages and the edges between the nodes correspond to the challenges sent by the verifier. Each node at depth
  $i$ has exactly $k_i$ children corresponding to the $k_i$ pairwise distinct challenges. Every transcript is represented by exactly
  one path from the root of the tree to a leaf node. 
\end{definition}

For brevity, we also write $\mathbf{k} = (k_1, k_2, \ldots, k_n) \in {\mathbb{N}}^n$ and refer to the tree of transcripts above
as $\mathbf{k}$-tree of transcripts.


\begin{definition} [$(k_1, k_2, \ldots, k_n)$-out-of-$(N_1, N_2, \ldots, N_n)$ Special-Soundness~\cite{EPRINT:AttFehKlo21}] \label{def:k-soundness}

  Let $k_1, k_2, \ldots, k_n, N_1, N_2, \ldots, N_n \in \mathbb{N}$. A $(2n + 1)$ move public-coin $\mathsf{PoK}$ protocol,
  where $\verifier$ samples the $i$-th challenge from a set of cardinality $N_i \geq k_i$ for $i \in [n]$, is
  $(k_1, k_2, \ldots, k_n)$-out-of-$(N_1, N_2, \ldots, N_n)$ special-sound if there exists a polynomial time algorithm that
  on an input statement $x$ and a $(k_1, k_2, \ldots, k_n)$-tree of accepting transcripts outputs a witness $w$. We also say
  $\mathsf{PoK}$ is $(k_1, k_2, \ldots, k_n)$-special-sound.

\end{definition}

Fiat and Shamir~\cite{FS} famously showed how to convert an interactive 
zero-knowledge PoK system into a non-interactive (here after NIZK-PoK) 
one in the random oracle model (ROM). 
The key observation was, if the verifier's challenges are 
sampled uniformly at random within their associated challenge spaces
(this is also called public-coin),
then one can emulate the verifier's behavior using a hash function modeled as a random oracle.
Thus, one can transform an interactive proof of knowledge into a (non-interactive) signature.

\begin{theorem}[Theorem 3,~\cite{EPRINT:AttFehKlo21} (Informal, rephrased)] \label{thm:FS-PoK}
	Fiat-Shamir transform of $(k_1, k_2, \ldots, k_n)$-out-of-$(N_1, N_2, \ldots, N_n)$
	special-sound interactive proof system is knowledge sound.
\end{theorem}

This means that special-sound, HVZK interactive proof systems can be
transformed into signature schemes using Fiat-Shamir transform.

In~\cite{FOCS:AmbRosUnr14,EC:Unruh16} it was shown that
the classical definition of commitment scheme (presented above)
 does not provide security (w.r.t. binding) in the presence of
  efficient quantum adversaries.
\footnote{The binding property guarantees that an efficient classical adversary
cannot produce \emph{two valid} openings for a given commitment such that
the openings can be used to open two distinct messages (say $m \neq m'$).
However, the situation is not so straightforward in the quantum setting. 
In~\cite{FOCS:AmbRosUnr14,EC:Unruh16} it was shown that it is possible
to construct an quantum adversary which produces 
a quantum state as commitment which is in superposition. And later
produce a \emph{single valid opening} for a message of its choice.
Clearly this is insecure since this quantum adversary can open the given
commitment to any message, but this is not captured by the classical binding
property because the attacker only produced a single valid opening.}
Unruh~\cite{EC:Unruh16} presented the suitable definition of
commitment scheme which is secure even when malicious sender possesses
efficient quantum computing capabilities, such post-quantum secure
commitment schemes are called \emph{collapse-binding commitments}.

We refer the interested readers
to~\cite{EC:Unruh16,AC:Unruh17,C:DFMS19,C:DonFehMaj20}
for an overview of relevant results  and 
the definitions of collapse-binding commitments,
quantum random oracle, collapse-binding hash functions 
(quantum analogue of collision-resistance)
which are crucial to the post-quantum security of our scheme.

In this paper, we instantiate the commitment scheme using
collapse-binding hash functions with appropriate salt values (randomness)
and the opening information simply reveals the salt.

\subsection{Stern Protocol} \label{sec:preliminaries:stern}

In this section, we present the initial Stern protocol \cite{Stern93} and discuss informally its correctness, soundness and zero-knowledge properties.
Given a public key $\pk = (\bm{H}, \bm{y}^\top)$, the Stern protocol permits to prove the knowledge of $\sk = (\bm{x})$ such that $\bm{y}^\top = \bm{H} \bm{x}^\top$ and $\hw{\bm{x}} = w$.
As such, the Stern protocol constitutes a proof of knowledge of a solution to an $\SD$ problem instance.
The idea of the protocol is to prove the knowledge of $\bm{x}$ using $\bm{x} + \bm{u}$ for some random value $\bm{u}$ and to prove that $\hw{\bm{x}} = w$ using $\pi[\bm{x}]$ for some random permutation $\pi$.

The protocol (see Figure \ref{fig:sd-plain1}) is a 3-rounds construction in which the prover $\prover$ starts by generating three commitments $c_1, c_2, c_3$ related to $\pi, \bm{x}$ and $\bm{u}$.
Then, the verifier $\verifier$ computes a challenge $\Ch \sampler \{0, 1, 2\}$ and the prover $\prover$ outputs a response $\Rsp$.
During the verification step, the verifier $\verifier$ recomputes two out of three commitments (depending on the value of the challenge $\Ch$) using $\Rsp$ and either $\accept$ or $\reject$ the proof.
The correctness of the proof follows straightforwardly from its description.

Regarding the soundness of the proof, one can see that a malicious prover $\tilde{P}$ can always cheat in two out of three cases.
Indeed, by choosing $\bm{z}$ such that $\bm{H} \bm{z} = \bm{H} \bm{u} + \bm{y}$ (no constraint on the weight of $\bm{z}$) and computing $c_3 = \cmt{\pi[\bm{z}]}$ and $\Rsp = (\pi, \bm{z})$ in the case $\Ch = 1$, the malicious prover $\tilde{\prover}$ will be accepted for both $\Ch = 0$ and $\Ch = 1$.
Similarly, by choosing $\tilde{\bm{x}}$ such that $\hw{\tilde{\bm{x}}} = w$ and computing $c_3 = \cmt{\pi[\tilde{\bm{x}} + \bm{u}]}$ and $\Rsp = (\pi[\tilde{\bm{x}}], \pi[\bm{u}])$ in the case $\Ch = 2$, the malicious prover $\tilde{\prover}$ will be accepted for both $\Ch = 0$ and $\Ch = 2$.
Finally, by choosing $\tilde{\bm{x}}$ such that $\hw{\tilde{\bm{x}}} = w$ and computing $c_1 = \cmt{\bm{H}(\tilde{\bm{x}} + \bm{u}) - \bm{y}}$ and $c_3 = \cmt{\pi[\tilde{\bm{x}} + \bm{u}]}$ as well as $\Rsp = (\pi, \tilde{\bm{x}} + \bm{u})$ in the case $\Ch = 1$ and $\Rsp = (\pi[\tilde{\bm{x}}], \pi[\bm{u}])$ in the case $\Ch = 2$, the malicious prover $\tilde{\prover}$ will be accepted for both $\Ch = 1$ and $\Ch = 2$.
In addition, if a malicious prover is able to cheat in the three cases, then one can build a knowledge extractor solving an instance of the $\SD$ problem from its response.
This explains why this protocol has a soundness error equal to $2/3$. 
In practice, one need to execute the protocol $\delta$ times to achieve a negligible soundness error.
This is not depicted in Figures~\ref{fig:sd-plain1} and \ref{fig:sd-structured1} which serve an illustrative purpose nonetheless these $\delta$ iterations are depicted in Figure \ref{fig:sd-structured2} which is used to describe a complete instantiation of our new construction.

Regarding the zero-knowledge property of the proof, one can see that if $\COM$ is a hiding commitment scheme, then $\Cmt$ does not reveal anything on the secret $\bm{x}$.
Moreover, the response of the prover $\prover$ only uses the secret through $\bm{u} + \bm{x}$ or $\pi[\bm{x}]$ namely $\bm{x}$ is either masked by a random vector $\bm{u}$ or random permutation $\pi$.
Formally, one can build a simulator that generates the view of an honest verifier with access to the public key only.

\input{pok_fig_plain1.tex}

\fi

\section{Quasi-Cyclic Stern Proof of Knowledge} \label{sec:protocol}
Our new PoK is based on the Stern protocol along with quasi-cyclicity and shares similarities with AGS \cite{AGS11}.
While one may think that Véron based protocols such as AGS would inherently be more efficient than Stern based ones, we prove this belief to be erroneous.
We introduce the non optimized QC Stern protocol in Section \ref{sec:protocol:simple}.
Then, we recall some optimizations from the literature and present new ones in Section \ref{sec:optimization}.
Finally, we describe the optimized QC Stern protocol and discuss its security in Section \ref{sec:protocol:optimized}.

\subsection{Quasi-Cyclic Stern Protocol} \label{sec:protocol:simple}
Our new protocol (see Figure~\ref{fig:sd-structured1}) is a ZK PoK for the Quasi-Cyclic Syndrome Decoding ($\QCSD$) problem.
Given inputs $(\bm{H}, \bm{y}) \sampler \mathcal{QC}(\Ft^{\kttk}) \times \Ft^{k}$, it allows a prover to convince a verifier that he knows $\bm{x} \in \Ft^{2k}$ such that $\bm{H}\bm{x}^{\top} = \bm{y}^{\top}$ and $\hw{\bm{x}} = w$ without revealing anything on it.
Let $\bm{a} = (a_0, \cdots, a_{k - 1}) \in \Ftk$, we define the $\rot{}{}$ operator as $\rot{\bm{a}}{r} := (a_{k - r + 1}, \cdots, a_{k - r})$.
For $\bm{b} = (\bm{b}_1, \bm{b}_2)\in \Ft^{2k}$, we slightly abuse notations and define $\rot{\bm{b}}{r} := (\rot{\bm{b}_1}{r}, \rot{\bm{b}_2}{r})$.
As we are considering QC matrices, one can see that $\bm{y}^\top = \bm{H} \bm{x}^\top \Leftrightarrow \rot{\bm{y}}{r}^\top = \bm{H} \, \rot{\bm{x}}{r}^\top$
hence one can prove the knowledge of the secret $\bm{x}$ associated to the public value $\bm{y}$ using any of $k$ different equations.
This property allows to introduce a new kind of challenge that permits to reduce the soundness error to $1/2$ asymptotically.

\input{pok_fig_structured1.tex}

\subsection{Existing Improvements and New Optimizations} \label{sec:optimization}

\vspace{0.5\baselineskip}
\noindent \textbf{Reducing the number of commitments} \cite{AGS11}\textbf{.}
We recall that code-based proofs of knowledge generally feature a soundness error of $2/3$ or $1/2 $ therefore one need to perform $\delta$ iterations of these protocols to get a negligible soundness error.
Using Figure~\ref{fig:sd-structured1} as an illustrative example, the idea is to compress the commitments over all the iterations so that only two initial commitments $\Cmt_1 = (c_{1,1} || c_{1,2} || \cdots || c_{\delta,1} || c_{\delta,2})$ and $\Cmt_2 = (c_{1,3} || \cdots || c_{\delta,3})$ need to be sent.
As the verifier is only able to reconstruct 2 out of 3 commitments himself, the prover must send him the missing commitment at the end of each iteration.
Overall, this reduces the number of commitments to be sent from $3\delta$ to $2 + \delta$. 

\vspace{0.5\baselineskip}
\noindent \textbf{Small weight vector compression} \cite{AGS11}\textbf{.}
The prover must reveal a permutation of a small weight vector $\pi[\bm{x}_r]$ to answer some challenges.
Leveraging the small weight of $\bm{x}_r$, one can compress $\pi[\bm{x}_r]$ before sending it thus reducing the cost of sending small weight vectors from $n$ to approximately $n/2$.

\vspace{0.5\baselineskip}
\noindent \textbf{Mitigation of an attack against 5-rounds protocols} \cite{ISIT21}\textbf{.} 
The attack against 5-rounds PoK from \cite{KZ20} is relevant for our construction.
The key idea of this attack is to split the attacker work in two steps by first trying to guess the first challenge for several repetitions and then guess the second challenge for the remaining repetitions.
One has to increase the number of iterations $\delta$ of the underlying PoK to ensure that the resulting signature remains secure, which increases its size.
One way to mitigate this attack is to consider $s$ instances of the $\SD$ problem within the keypair namely using $\sk = (\bm{x}^i)_{i \in [1,s]}$ and $\pk = (\bm{H}, (\bm{y}^i)^\top = \bm{H}(\bm{x}^i)^\top)_{i \in [1,s]}$.
Doing so, the first challenge space size is increased from $k$ to $s \times k$ which makes the attack less efficient so that $\delta$ won't need to be increased as much as initially thought.
In practice, this introduces a trade-off between key size and signature size.

\vspace{0.5\baselineskip}
\noindent \textbf{Additional vector compression from seeds.}
Starting from the initial Stern proposal, all constructions suggest to use seeds to reduce communication costs.
Using Figure \ref{fig:sd-structured1} for illustrative purposes, one can use a seed $\theta$ to compute $\pi$ and then substitute $\pi$ by $\theta$ in the prover's response.
We now introduce an additional vector compression that is only applicable to $\SD$ based protocols.
One can go one step further and use a seed $\xi$ to generate a random value $\bm{v} \samples{\xi} \Ftn$ and then compute the value $\bm{u} = \pi^{-1}[\bm{v}]$.
When the prover is required to send $\pi[\bm{u}]$, he now needs to send $\bm{v}$ which can be substituted by $\xi$.
Instead of sending a vector of size $n$, the prover only sends a seed which greatly reduces the communication cost.
We now explain why this optimization can't be applied to the AGS protocol.
Under the $\GSD$ representation, the prover need to send $\pi[\bm{uG}]$ rather than $\pi[\bm{u}]$ whenever $\Ch_2 = 1$.
But the quantity $\pi[\bm{uG}]$ cannot be replaced by a seed
generating it as $\bm{uG}$ is a codeword hence the optimization can not be applied.

\vspace{0.5\baselineskip}
\noindent \textbf{Seed and commitment compression.} 
Using the previous optimization, one can see that the prover sends one seed during each iteration either $\theta$ from which $\pi$ can be recomputed or $\xi$ from which $\pi[\bm{u}]$ can be recomputed.
Let us consider two consecutive iterations of the protocol, the prover will have to send one of the following tuple of seeds: $(\theta_1, \theta_2), (\theta_1, \xi_2), (\xi_1, \theta_2), (\xi_1, \xi_2)$.
If master seeds $\theta$ (respectively $\xi$) are used to generate $\theta_1$ and $\theta_2$ (respectively $\xi_1$ and $\xi_2$), then the prover will have to send one of the following values: $\theta, (\theta_1, \xi_2), (\xi_1, \theta_2), \xi$.
By using such a technique, one reduces the average communication cost associated to seeds by $25$\%.
This optimization can be seen as a variation of the seed compression optimization from \cite{KKW18} in which the (unique) binary tree used is replaced by several binary trees of depth~1.
Similarly, one can also group commitments using binary trees of depth~1 (from bottom to top contrarily to the previous case) which reduces the cost associated to commitments from $(2 + \delta) \cdot |\com|$ to $(2 + 0.75\delta) \cdot |\com|$.

\subsection{Optimized Quasi-Cyclic Stern Protocol} \label{sec:protocol:optimized}

Figure~\ref{fig:sd-structured2} describes the optimized version of QC Stern protocol and includes the $\delta$ iterations required to reduce the soundness bellow $2^{-\lambda}$ where $\lambda$ is the security parameter.
As it is usually done, we don't include binary tree related optimizations nor small weight vector related optimizations as this simplifies the description of the protocol while not being related to its security.
We discuss the soundness and zero-knowledge properties of our PoK giving only sketches of proof and defer the reader to the full version of the paper for additional details \cite{ISIT22-long}.
The soundness relies on a reduction from the $\QCSD$ problem to the $\DSD$ problem. 

\input{pok_fig_structured2.tex}

\input{pok_correctness.tex}

\input{pok_soundness.tex}

\input{pok_hvzk.tex}


\section{Parameters and Resulting Sizes} \label{sec:parameters}

The protocol described in Figure~\ref{fig:sd-structured2} can be turned into
a signature using the Fiat-Shamir transform \cite{FS, DPJS96, ExtendedFS, ExtendedFS2, EPRINT:AttFehKlo21}. Hereafter, we discuss
the choice of our parameters and compare our protocol with existing schemes.

\vspace{0.5\baselineskip}
\noindent \textbf{Decoding attack.} We consider the BJMM generic decoder~\cite{BJMM12} 
  with estimates from \cite{HS13}. The parameters
  $(n,k,w)$ are chosen such that decoding $w$ errors in a binary
  quasi-cyclic $[n,k]$ code costs at least $2^\lambda$. The attacker
  has access to $N=sk$ syndromes ($k$ rotations of $s$ public
  keys) and is successful by decoding only one of them. As shown in
  \cite{Sen11}, this multiple targets attack reduces the complexity by
  a factor at most $\sqrt{N}$. 

\vspace{0.5\baselineskip}
\noindent \textbf{Soundness error.} Following \cite{AGS11}, for given $(n,k,w)$,
  solving $\DSD(n,k,w,\alpha)$ provides a solution to $\QCSD(n,k,w)$
  with probability $1-\varepsilon(\alpha)$ where
  $\varepsilon(\alpha)\approx{n\choose
    w}^{\alpha-1}/2^{(n-k)(\alpha-2)}$.
  The soundness error for one iteration cannot exceed
  $\rho^*=\frac{sk+\alpha^*-1}{2sk}$ where $\alpha^*$ is the largest
  integer such that $\varepsilon(\alpha^*)\le2^{-\lambda}$. The
  soundness error for $\delta$ iterations is $(\rho^*)^\delta$ and it
  is lower than $2^{-\lambda}$ if
  $\delta\ge\frac{-\lambda}{\log_2\rho^*}$.

\vspace{0.5\baselineskip}
\noindent \textbf{Attack against 5-rounds protocols.} The attack of
\cite{KZ20} can be used against our protocol. For $\delta$ iterations 
of the protocol, the attacker will find the value of $\tau^*$ (the number 
of second challenges to guess) which minimizes the attack cost 
$P^{-1} + 2^{\delta-\tau^*}$ where
$P=\sum_{\tau\ge\tau^*}{\delta\choose
  \tau}\left(\frac{1}{sk}\right)^\tau\left(\frac{sk-1}{sk}\right)^{\delta-\tau}$. The
choice of $\delta$ must be such that this cost is $\ge 2^\lambda$.

\vspace{0.5\baselineskip}
\noindent \textbf{Signature length and scalability.} The signature
consists of the outputs of $\prover_1, \prover_2, \prover_3$ namely two 
commitments and a seed along with all the $d_i$ (see
Figure~\ref{fig:sd-structured2}). Each response $d_i$ consists of a
seed, a commitment, and a word of $\Ftn$ with no particular structure
if $b_i=0$ and of weight $w$ if $b_i=1$.  The seeds and commitments are taken of
length $\lambda$ and $2\lambda$ respectively. The
words of weight $w$ can be compressed to $n-k$ bits. Finally the
seeds and commitments can be structured pairwise as explained in
Section~\ref{sec:optimization}, allowing to save one seed and one commitment every 
4 iterations on average. For codes of rate $1/2$ ($k=n/2$) the average length of
the signature is $|\sigma| = 5 \lambda + \delta(0.75n + 2.25\lambda)$. Both
$\delta$ and $n$ will grow linearly with the security parameters
$\lambda$ and thus the signature length grows as $\lambda^2$, roughly
we have here $|\sigma|\approx11\lambda^2$.

\begin{table}[!htb]
  \caption{Parameters and signature sizes in bytes for $\lambda = 128$} \label{table:param1}
  \vspace{-0.5\baselineskip}
  \begin{center}
    {\scriptsize {\renewcommand{\arraystretch}{1.3}
    \begin{tabular}{|c|c|c|c|c|c|c|c|}
        \hline
        $n$ & $k$ & $w$ & $\delta$ & $s$ & $\sk$ size & $\pk$ size & $\sigma$ size \\ \hline
        \multirow{3}{*}{~1306~} & \multirow{3}{*}{~653~} & \multirow{3}{*}{~137~} & 151 &  1 & 16 B & 0.1 kB & 24.1 kB \\ \cline{4-8}
                                                                              & & & 145 &  4 & 16 B & 0.4 kB & 23.1 kB \\ \cline{4-8}
                                                                              & & & 141 & 20 & 16 B & 1.7 kB & 22.5 kB \\ \hline
    \end{tabular}}}
  \end{center}
\vspace{-0.5\baselineskip}
\end{table}

\vspace{0.5\baselineskip}
\noindent \textbf{Comparison with code-based schemes.} 
We compare our proposal to code-based signatures constructed from PoK for the SD problem in Table~\ref{table:param2}.
We consider both size and expected performances as criteria.
%
\iflong
Since there is no implementation available for most of these schemes yet, we provide an estimate of their expected relative performances following the methodology and parameters from \cite{BGKM22}.
For all these schemes, the first step (every operations executed by the prover before he outputs its first commitment) is likely to dominate the overall performance cost.
This step can be seen as repeating $\mu$ times the computation of $\nu$ operations whose cost is approximated to be equal amongst schemes.
This introduces an approximation in our comparison which could only be solved by providing and benchmarking actual implementations of the schemes. 
In particular, this approximation hides the performance difference between using plain matrices and structured ones which is not negligible in practice. 
As such, one should compare the schemes involving a regular matrix / vector multiplication separately from the schemes that have a more efficient one thanks to the use of structured matrices.
Moreover, the GPS scheme does not include such a multiplication in this step hence the proposed estimate might overestimate its real cost.
%
\fi
%
Overall, one can see that our proposal offers an interesting trade-off between cost and sizes as it has the smallest expected cost while still featuring competitive sizes.

In addition, we also present in Table~\ref{table:param3} the sizes of other code-based signatures.
Wave~\cite{wave19} is based on the SD problem over $F_3$ (with secrets of large weights) and the Generalized $(U, U+V)$-codes indistinguishability, LESS~\cite{less} relies on the permutation code equivalence problem and Durandal~\cite{Durandal} is based on the rank SD problem and the product spaces subspaces indistinguishability.

\begin{table}[!htb]
  \caption{Signatures from PoK for the SD problem ($\lambda = 128$)} \label{table:param2}
  \vspace{-0.5\baselineskip}
  \begin{center}
    {\scriptsize{\renewcommand{\arraystretch}{1.3}
    \begin{tabular}{|l|c|c|c|c|c|}
      \cline{2-6}
      \multicolumn{1}{c|}{}  & \multicolumn{3}{c|}{Performance} & \multicolumn{2}{c|}{Size} \\ \cline{2-6}
      \multicolumn{1}{c|}{}  & $\mu$ & $\nu$ & Cost & $\pk$ & $\sigma$ \\ \hline 
      Stern \cite{Stern93}                         & 219  & 2    & 438       & 0.1 kB                   & 36.2 kB                   \\ \hline 
      Véron \cite{Veron97, ISIT21}                 & 219  & 2    & 438       & 0.2 kB                   & 30.8 kB                   \\ \hline 
      CVE \cite{CVE11}                             & 156  & 2    & 312       & 0.3 kB                   & 31.4 KB                   \\ \hline
      \multirow{3}{*}{AGS \cite{AGS11}}            & 151  & 2    & 302       & 0.2 kB                   & 29.3 kB                   \\ \cline{2-6} 
                                                   & 145  & 2    & 290       & 0.7 kB                   & 28.2 kB                   \\ \cline{2-6}
                                                   & 141  & 2    & 282       & 3.3 kB                   & 27.4 kB                   \\ \hline
      \multirow{2}{*}{GPS \cite{GPS21}}            & 512  & 128  & 65 536    & 0.2 kB                   & 27.1 kB                   \\ \cline{2-6}
                                                   & 4096 & 1024 & 4 194 304 & 0.2 kB                   & 19.8 kB                   \\ \hline
      \multirow{2}{*}{FJR \cite{FJR21}}            & 187  & 8    & 1496      & 0.1 kB                   & 24.4 kB                   \\ \cline{2-6}
                                                   & 389  & 32   & 12 448    & 0.1 kB                   & 17.6 kB                   \\ \hline
      \multirow{1}{*}{BGKM (Sig. 1) \cite{BGKM22}} & 256  & 2    & 512       & 0.1 kB                   & 24.3 kB                   \\ \hline
      \multirow{3}{*}{This Paper}                  & 151  & 2    & 302       & 0.1 kB                   & 24.1 kB                   \\ \cline{2-6} 
                                                   & 145  & 2    & 290       & 0.4 kB                   & 23.1 kB                   \\ \cline{2-6}
                                                   & 141  & 2    & 282       & 1.7 kB                   & 22.5 kB                   \\ \hline 
    \end{tabular}}}
  \end{center}
\vspace{-\baselineskip}
\end{table}

\begin{table}[!htb]
  \caption{Other code-based signatures ($\lambda = 128$)} \label{table:param3}
  \vspace{-0.5\baselineskip}
  \begin{center}
    {\scriptsize{\renewcommand{\arraystretch}{1.3}
    \begin{tabular}{|l|c|c|c|}
      \cline{2-4}
      \multicolumn{1}{c|}{} & $\pk$ & $\sigma$ & $\pk + \sigma$ \\ \hline 
      Wave \cite{wave19}       & 3.2 MB  & 0.93 kB  & 3.3 MB   \\ \hline
      LESS \cite{less}         & 11.6 kB & 10.4 kB  & 22.0 kB  \\ \hline
      Durandal \cite{Durandal} & 15.3 kB &  4.1 kB  & 19.4 kB  \\ \hline
    \end{tabular}}}
  \end{center}
\vspace{-0.5\baselineskip}
\end{table}

\vspace{0.5\baselineskip}
\noindent \textbf{Comparison with other schemes.} 
Outside of code-based cryptography, there exist many other signatures based on the Fiat-Shamir transform.
Some of them were submitted to the NIST standardization process \cite{picnic,MQDSS} while other have been published recently \cite{beullens2020}.
All these schemes reduce to a given difficult problem like for instance the MQ or PKP problems. 
For 128 bits of security, depending on these different post-quantum ZK signature schemes, the size of the signature may vary (also depending on chosen trade-offs) between 12kB for~\cite{beullens2020} and 40kB for MQDSS \cite{MQDSS}. 
A strong feature of the SD problem being that the problem has been
used for a long time, the attacks are well understood, and thus
significant speedups in the attacks are unlikely to occur.

\iflong
\section{Generalization to Additional Metrics} \label{sec:generalization}

The protocol described in this paper can be generalized to other type of weight through the use of Full Domain Linear Isometries such as Hamming weight (classical or large) over $\Fq$ or rank weight.
We define a Full Domain Linear Isometry (FDLI) as a set of linear isometries which has the property that given a random element $f$ of this set, the image by $f$ of a random word $\bm{x}$ of weight $w$ is a random word $\bm{y}$ of weight $w$.
Hereafter, we give some examples:

\begin{itemize}

  \item[$\bullet$] Hamming weight and words over $\Ftn$ (the case of this paper): consider the set of permutations $\hperm$ ;

  \item[$\bullet$] Hamming weight and words over $\Fq^n$: consider the set of monomial matrices (a permutation matrix times a diagonal matrix with non null elements of $\Fq$ on its diagonal). It can work for classical Hamming weight but also for more recent large weight \cite{bricout2019ternary} ;

  \item[$\bullet$] Rank weight and vectors over $\Fqmn$: consider the function (originally defined in \cite{RankStern11}) which associate to a word $\bm{x}$ of $\Fqmn$ a $m \times n$ matrix $\bm{X}$ by writing elements of $\Fqm$ as vectors of $\Fq^m$ through a basis $\beta$ of $\Fqm$ over $\Fq$. 
    Then multiply $\bm{X}$ by two random invertible matrices $\bm{P}$ and $\bm{Q}$ over $\Fq$ of respective sizes $m \times m$ and $n \times n$ such that $\bm{X} \rightarrow \bm{PXQ}$. 
    Finally, rewrite $\bm{PXQ}$ as an element of $\Fqmn$ through the basis $\beta$.
    $\bm{P}$ permits to associate a support (a vector space) of weight $w$ to any other support of weight $w$ and the matrix $\bm{Q}$ permits to associate words with the same support.  
\end{itemize}

When considering such FDLI functions the protocol described in this paper can directly be adapted for a given weight, simply by changing the weight and replacing the function $\pi$ in the protocol by a random element of a FDLI set.
The proof are straightforward and are the same than in the binary Hamming case described in this paper; we omit them in this short version of the paper. 
The notion of linearity and isometry permits to prove the correctness and the soundness of the adapted protocol while the full domain property permits to prove the zero-knowledge property.
Interestingly enough, there was an adaptation of the Stern protocol in \cite{Chen95} where the author had kept permutations but in a rank metric context.
In this case, the support of the secret vector was not modified and led to a (unseen at that time) break of zero-knowledge and ultimately a break of the whole protocol as described in \cite{RankStern11}.
\fi

\newpage
\bibliographystyle{splncs04}
\bibliography{ref}

\end{document}

%% file: pok_fig_plain1.tex
\begin{figure}[!htb]
  \begin{center}
  \resizebox{0.44\textwidth}{!}{\fbox{
    \pseudocode{%
      \hspace{0pt} \\[-0.75\baselineskip] 
      \underline{\setup(1^\lambda) ~\&~ \keygen(\param)} \\
      \param = (n, k, w) \samplen \setup(1^{\lambda}) \\
      \bm{x} \sampler \Ftn \text{ such that } \hw{\bm{x}} = w \\
      \bm{H} \sampler \Ft^{\nmktn}, ~ \bm{y}^{\top} = \bm{H} \bm{x}^{\top} \\
      (\sk, \pk) = (\bm{x}, (\bm{H}, \bm{y})) \\[0.75\baselineskip]
      \underline{\prover_1(\param, \sk, \pk)} \\
      \pi \sampler \hperm, ~ \bm{u} \sampler \Ftn \\
      c_{1} = \cmt{\pi \, || \, \bm{H} \bm{u}^{\top}} \\
      c_{2} = \cmt{\pi[\bm{u}]} \\
      c_{3} = \cmt{\pi[\bm{u} + \bm{x}]} \\
      \Cmt = (c_{1}, c_{2}, c_{3}) \\[0.75\baselineskip]
      \underline{\verifier_1(\param, \pk, \Cmt)} \\
      \Ch \sampler \{0, 1, 2\} \\[0.75\baselineskip]
      \underline{\prover_2(\param, \sk, \pk, \Cmt, \Ch)} \\
      \pcif \Ch = 0 \pcthen \\
      \pcind \Rsp = \big( \pi, \, \bm{u} \big) \\
      \pcif \Ch = 1 \pcthen \\
      \pcind \Rsp = \big( \pi, \, \bm{u} + \bm{x} \big) \\
      \pcif \Ch = 2 \pcthen \\
      \pcind \Rsp = \big( \pi[\bm{u}], \, \pi[\bm{x}] \big) \\[0.75\baselineskip]
      \underline{\verifier_2(\param, \pk, \Cmt, \Ch, \Rsp)} \\
      \pcif \Ch = 0 \pcthen \\
      \pcind \bar{c}_{1} = \cmt{\pi \, || \, \bm{H} \bm{u}^{\top}} \\
      \pcind \bar{c}_{2} = \cmt{\pi[\bm{u}]} \\[0.5\baselineskip]
      \pcind \pcif c_1 \neq \bar{c}_1 \pcor c_2 \neq \bar{c}_{2} \pcthen \\
      \pcind \pcind \pcreturn \reject \\[0.75\baselineskip]
      \pcif \Ch = 1 \pcthen \\
      \pcind \bar{c}_{1} = \cmt{\pi \, || \, \bm{H}(\bm{u} + \bm{x})^{\top} - \bm{y}^{\top}}  \\
      \pcind \bar{c}_{3} = \cmt{\pi[\bm{u} + \bm{x}]} \\[0.5\baselineskip]
      \pcind \pcif c_1 \neq \bar{c}_1 \pcor c_3 \neq \bar{c}_{3} \pcthen \\
      \pcind \pcind \pcreturn \reject \\[0.75\baselineskip]
      \pcif \Ch = 2 \pcthen \\
      \pcind \bar{c}_{2} = \cmt{\pi[\bm{u}]}  \\
      \pcind \bar{c}_{3} = \cmt{\pi[\bm{u}] + \pi[\bm{x}]} \\[0.5\baselineskip]
      \pcind \pcif c_2 \neq \bar{c}_2 \pcor c_3 \neq \bar{c}_{3} \pcor \hw{\pi[\bm{x}]} \neq w \pcthen \\
      \pcind \pcind \pcreturn \reject \\[0.75\baselineskip]
      \pcreturn \accept
    }}}
  \caption{Stern Protocol} \label{fig:sd-plain1} 
  \end{center}
\end{figure}

%% file: pok_fig_structured1.tex
\iflong
\newcommand\boxresize{0.53}
\newcommand\hackscale{0pt}
\fi

\begin{figure}[!htb]
  \begin{center}
  \resizebox{\boxresize\textwidth}{!}{\fbox{
    \pseudocode{%
      \hspace{\hackscale} \\[-0.75\baselineskip] 
      \underline{\setup(1^\lambda) ~\&~ \keygen(\param)} \\
      \param = (k, w) \samplen \setup(1^{\lambda}) \\
      \bm{x} \sampler \Ft^{2k} \text{ such that } \hw{\bm{x}} = w \\
      \bm{H} \sampler \mathcal{QC}(\Ft^{\kttk}), ~ \bm{y}^{\top} = \bm{H} \bm{x}^{\top} \\
      (\sk, \pk) = (\bm{x}, (\bm{H}, \bm{y})) \\[0.5\baselineskip]
      \underline{\prover_1(\param, \sk, \pk)} \\
      \pi \sampler \hpermtk, ~ \bm{u} \sampler \Ft^{2k} \\
      c_{1} = \cmt{\pi \, || \, \bm{H} \bm{u}^{\top}}, ~ c_{2} = \cmt{\pi[\bm{u}]} \\
      \Cmt_1 = (c_{1}, c_{2}) \\[0.5\baselineskip]
      \underline{\verifier_1(\param, \pk, \Cmt_1)} \\
      r \sampler [0, k - 1] \\
      \Ch_1 = r \\[0.5\baselineskip]
      \underline{\prover_2(\param, \sk, \pk, \Cmt_1, \Ch_1)} \\
      \bm{x}_{r} = \rot{\bm{x}}{r}, ~ c_{3} = \cmt{\pi[\bm{u} + \bm{x}_{r}]} \\
      \Cmt_2 = c_{3} \\[0.5\baselineskip]
      \underline{\verifier_2(\param, \pk, \Cmt_1, \Ch_1, \Cmt_2)} \\
      \Ch_2 \sampler \bit  \\[0.5\baselineskip]
      \underline{\prover_3(\param, \sk, \pk, \Cmt_1, \Ch_1, \Cmt_2, \Ch_2)} \\
      \pcif \Ch_2 = 0 \pcthen \\
      \pcind \Rsp = \big( \pi, \, \bm{u} + \bm{x}_{r} \big) \\
      \pcif \Ch_2 = 1 \pcthen \\
      \pcind \Rsp = \big( \pi[\bm{u}], \, \pi[\bm{x}_{r}] \big) \\[0.5\baselineskip]
      \underline{\verifier_3(\param, \pk, \Cmt_1, \Ch_1, \Cmt_2, \Ch_2, \Rsp)} \\
      \pcif \Ch_2 = 0 \pcthen \\
      \pcind \bar{c}_{1} = \cmt{\pi \, || \, \bm{H}(\bm{u} + \bm{x}_{r})^{\top} - \rot{\bm{y}}{r}^{\top}}  \\
      \pcind \bar{c}_{3} = \cmt{\pi[\bm{u} + \bm{x}_{r}]} \\
      \pcind \pcif c_1 \neq \bar{c}_1 \pcor c_3 \neq \bar{c}_{3} \pcthen \\
      \pcind \pcind \pcreturn \reject \\[0.5\baselineskip]
      \pcif \Ch_2 = 1 \pcthen \\
      \pcind \bar{c}_{2} = \cmt{\pi[\bm{u}]}  \\
      \pcind \bar{c}_{3} = \cmt{\pi[\bm{u}] + \pi[\bm{x}_{r}]} \\
      \pcind \pcif c_2 \neq \bar{c}_2 \pcor c_3 \neq \bar{c}_{3} \pcor \hw{\pi[\bm{x}_{r}]} \neq w \pcthen \\
      \pcind \pcind \pcreturn \reject \\
      \pcreturn \accept
    }}}
  \caption{Quasi-Cyclic Stern Protocol (one iteration)} \label{fig:sd-structured1}
  \vspace{-\baselineskip}
  \end{center}
\end{figure}

%% file: pok_fig_structured2.tex
\iflong
\newcommand\boxresizeb{0.48}
\newcommand\hackscaleb{0pt}
\fi

\begin{figure}[!htb]
  \begin{center}
  \resizebox{\boxresizeb\textwidth}{!}{\fbox{
    \pseudocode{%
      \hspace{\hackscaleb} \\[-0.75\baselineskip] 
      \underline{\setup(1^\lambda) ~\&~ \keygen(\param)} \\
      \param = (k, w, \delta, s, |\seed|) \samplen \setup(1^{\lambda}) \\
      \phi_1 \sampler \bit^{|\seed|}, \phi_2 \sampler \bit^{|\seed|}, 
      \bm{H} \samples{\phi_2} \mathcal{QC}(\Ft^{\kttk}) \\
      \pcfor i \in \intoneto{s} \pcdo ~\{~ \bm{x}^i \samples{\phi_1} \Ft^{2k}, ~ (\bm{y}^i)^{\top} = \bm{H} (\bm{x}^i)^{\top} ~\} \\
      (\sk, \pk) = (\phi_1, (\phi_2, \bm{y}^1, \cdots, \bm{y}^s)) \\[0.5\baselineskip]
      \underline{\prover_1(\param, \sk, \pk)} \\
      \pcfor i \in \intoneto{\delta} \pcdo \\
      \pcind \theta_i \sampler \bit^{|\seed|}, ~ \pi_i
      \samples{\theta_i} \hpermtk \\
      \pcind \xi_i \sampler \bit^{|\seed|}, ~ \bm{v}_i \samples{\xi_i} \Ft^{2k}, ~ \bm{u}_i = \pi^{-1}[\bm{v}_i]\\
      \pcind c_{i,1} = \cmt{\pi_i \, || \, \bm{H} \bm{u}_i^{\top}}, ~ c_{i,2} = \cmt{\pi_i[\bm{u}_i]} \\
      \Cmt_1 = \cmt{c_{1,1} || c_{1,2} || \cdots || c_{\delta,1} || c_{\delta,2}} \\[0.5\baselineskip]
      \underline{\verifier_1(\param, \pk, \Cmt_1)} \\
      \pcfor i \in \intoneto{\delta} \pcdo ~\{~ s_i \sampler [0, s - 1], ~ r_i \sampler [0, k - 1] ~\} \\
      \Ch_1 = ((s_1, r_1), \cdots, (s_\delta, r_\delta)) \\[0.5\baselineskip]
      \underline{\prover_2(\param, \sk, \pk, \Cmt_1, \Ch_1)} \\
      \pcfor i \in \intoneto{\delta} \pcdo \\
      \pcind \bm{x}^{s_i}_{r_i} = \rot{\bm{x}^{s_i}}{r_i}, ~ c_{i,3} = \cmt{\pi_i[\bm{u}_i + \bm{x}^{s_i}_{r_i}]} \\
      \Cmt_2 = \cmt{c_{1,3} || \cdots || c_{\delta,3}} \\[0.5\baselineskip]
      \underline{\verifier_2(\param, \pk, \Cmt_1, \Ch_1, \Cmt_2)} \\
      \pcfor i \in \intoneto{\delta} \pcdo ~\{~ b_i \sampler \bit ~\} \\
      \Ch_2 = (b_1, \cdots, b_\delta) \\[0.5\baselineskip]
      \underline{\prover_3(\param, \sk, \pk, \Cmt_1, \Ch_1, \Cmt_2, \Ch_2)} \\
      \pcfor i \in \intoneto{\delta} \pcdo \\
      \pcind \pcif b_i = 0 \pcthen ~ d_i = \big( \theta_i, \, \bm{u}_i + \bm{x}^{s_i}_{r_i}, \, c_{i,2} \big) \\
      \pcind \pcif b_i = 1 \pcthen ~ d_i = \big( \xi_i, \, \pi_i[\bm{x}^{s_i}_{r_i}], \, c_{i,1} \big) \\
      \Rsp = (d_1, \cdots, d_\delta) \\[0.5\baselineskip]
      \underline{\verifier_3(\param, \pk, \Cmt_1, \Ch_1, \Cmt_2, \Ch_2, \Rsp)} \\
      \pcfor i \in \intoneto{\delta} \pcdo \\
      \pcind \pcif b_i = 0 \pcthen \\
      \pcind \pcind \pi_i \samples{\theta_i} \hpermtk, ~ \bar{c}_{i,2} = c_{i,2} \\
      \pcind \pcind \bar{c}_{i,1} = \cmt{\pi_i \, || \, \bm{H}(\bm{u}_i + \bm{x}^{s_i}_{r_i})^{\top} - \rot{\bm{y}^{s_i}}{r_i}^{\top}}  \\
      \pcind \pcind \bar{c}_{i,3} = \cmt{\pi_i[\bm{u}_i + \bm{x}^{s_i}_{r_i}]} \\[0.5\baselineskip]
      \pcind \pcif b_i = 1 \pcthen \\
      \pcind \pcind \bm{v}_i \samples{\xi_i} \Ft^{2k}, ~ \bar{c}_{i,1} = c_{i,1} \\
      \pcind \pcind \bar{c}_{i,2} = \cmt{\bm{v}_i}, ~ \bar{c}_{i,3} = \cmt{\bm{v}_i + \pi_i[\bm{x}^{s_i}_{r_i}]} \\
      \pcind \pcind \pcif \hw{\pi_i[\bm{x}^{s_i}_{r_i}]} \neq w \pcthen \pcreturn \reject \\[0.5\baselineskip]
      \pcif \open{\Cmt_1, \, \bar{c}_{1,1} || \bar{c}_{1,2} || \cdots || \bar{c}_{\delta, 1} || \bar{c}_{\delta,2}} \neq 1 \pcthen \\
      \pcind \pcreturn \reject \\[0.25\baselineskip]
      \pcif \open{\Cmt_2, \, \bar{c}_{1,3} || \cdots || \bar{c}_{\delta,3}} \neq 1 \pcthen \\
      \pcind \pcreturn \reject \\[0.25\baselineskip]
      \pcreturn \accept
    }}}
  \caption{Quasi-Cyclic Stern Protocol (with optimizations)} \label{fig:sd-structured2} 
  \end{center}
  \vspace{-\baselineskip}
\end{figure}

%% file: pok_correctness.tex
\iflong
\begin{theorem}[Correctness] \label{thm:correctness}
  The proof of knowledge depicted in Figure \ref{fig:sd-structured2} satisfies the correctness property.
\end{theorem}

\begin{proof}
  The correctness follows straightforwardly from the protocol description.
\end{proof}
\fi

%% file: pok_soundness.tex
\vspace{0.5\baselineskip}
\begin{definition}[$\DSD$ problem] \label{def:diffsd}
   Given positive integers ($n=2k$), $k$, $w$, $\alpha$, a
   parity-check matrix of a quasi-cyclic code $\bm{H} \sampler
   \mathcal{QC}(\Ft^{\kttk})$ and $\bm{y} \in \Ft^{k}$ such
   that $\bm{Hx}^\top = \bm{y}^\top$ where $\bm{x} \in \Ft^{2k}$ and
   $\hw{\bm{x}} = w$. The Differential Syndrome Decoding problem
   $\DSD(n,k,w,\alpha)$ asks to find a set of vectors $(\bm{c},
   (\bm{z}_1, \cdots, \bm{z}_{\alpha})) \in \Ft^{k} \times
   (\Ft^{2k})^\alpha$ such that for each $i \in \intoneto{\alpha}$, 
   $\bm{H}\bm{z}_i^\top + \bm{c} = \bm{rot}_{i}(\bm{y}^\top)$  with $\hw{\bm{z}_i} = w$.
\end{definition}

\vspace{0.5\baselineskip}
\begin{theorem}[\textbf{$\QCSD$ to $\DSD$ reduction} \cite{AGS11, schrekPhD}] \label{thm:reduction}
  If there exists a Probabilistic Polynomial-Time ($\ppt$) algorithm solving the $\DSD(n,k,w,\alpha)$ problem with success probability $p$, 
  then there exists a $\ppt$ algorithm solving the $\QCSD(n,k,w)$ problem with success probability $(1 - \frac{{n \choose w}^{\alpha-1}}{2^{(n-k)(\alpha - 2)}}) \cdot p$.
\end{theorem}
\vspace{0.5\baselineskip}

\iflong
\begin{proof}
  The proof of Theorem \ref{thm:reduction} can be found in \cite{AGS11, schrekPhD}.
\end{proof}
\fi

The security of multi-round Fiat-Shamir transformation has been analyzed in \cite{EPRINT:AttFehKlo21}.
Following their definitions, we provide a soundness proof compatible with 5-round protocols.
Such a proof was lacking in previous quasi-cyclic based proposals.
A $(q,2)$-tree of transcripts for a $5$-round (public coin) protocol is a set of $2q$ transcripts arranged in a tree structure. 
The nodes in the tree represent the prover's messages and the edges between the nodes correspond to the verifier's challenges. 
%
\iflong
\begin{proof}
The root of the tree is the first prover's message and has exactly $q$ children corresponding to the $q$ pairwise distinct challenges.
Each node at depth $1$ has exactly $2$ children leaf nodes, which correspond to final responses to the challenge bit sent in the fourth round.
\end{proof}
\fi
%
Each transcript is represented by a path from the three root to a leaf node.
We say that the protocol is $(q, 2)$ special-sound if there exists a $\ppt$ algorithm that
on an input statement and a $(q,2)$-tree of accepting transcripts outputs a witness. 

\vspace{0.5\baselineskip}
\begin{theorem}[\textbf{(sk, 2)-Special Soundness}] \label{thm:soundness}
  Let $k$ and $\delta$ be public parameters denoting the dimension of a $[n=2k, k]$ QC code and the number of iterations within the protocol.
  If $\COM$ is a binding commitment scheme, then the PoK depicted in Figure~\ref{fig:sd-structured2} is sound with soundness error $(\frac{sk + \alpha - 1}{2sk})^\delta$ for some parameters $\alpha$ and $s$ under the $\QCSD$ assumption.
\end{theorem}
\vspace{0.5\baselineskip}

\iflong

\begin{proof}
  
  We first prove that if there is an adversary $\adv$ that is able to cheat with probability greater than $(\frac{sk + \alpha - 1}{2sk})^\delta$ then one can build a knowledge extractor $\ext$ that solve the $\DSD(n, k, w, \alpha)$ problem.
  Such an adversary is able to cheat with probability at least $(\frac{sk + \alpha}{2sk})$ in a given iteration $i$ where $i \in [1, \delta]$. 
  Within this specific iteration (we omit the index $i$ for simplicity hereafter), 
  we show that the protocol is $(sk,2)$-special-sound. In order to prove this we need to show that given $(sk,2)$-tree of accepting 
  transcripts, one can efficiently extract the witness $x$. Note that in order the prove $(sk,2)$-special-soundness, the extractor
  is given $2sk$ accepting transcripts, however our proof strategy requires only $sk + \alpha$ accepting transcripts (with pairwise distinct challenges
  in each round) for $\alpha \geq 1$. 
  We begin by first visualizing the $(sk,2)$-tree of accepting transcripts, note that each accepting transcript corresponds to a path from the root
  (which represents the first commitment message) to the leaf nodes. The root has $sk$ children corresponding to the first challenge message, and
  subsequently each of these children has exactly $2$ leaf nodes attached to them corresponding to the second challenge message.
 Since the extractor is given $sk + \alpha$ accepting transcripts, by pigeonhole principle there are exactly $\alpha$ nodes at first level such that
  the extractor knows both the leaf nodes attached to these intermediate nodes. However, this is sufficient for extracting witness already since knowledge
  of messages corresponding to \emph{both leaf nodes sharing the same parent intermediate node} implies that the extractor knows the prover's response
  for both challenge bits $b = 0$ and $b = 1$. We next show how this information is used to solve the $\DSD(n, k, w, \alpha)$ problem. 
  Let $((s_1, r_1), \cdots, (s_{\alpha}, r_{\alpha}))$ be the challenges for which the adversary $\adv$ is able to produce accepting transcripts for both the cases $b = 0$ and $b = 1$.
  Let $d^{\, r_j, b} = (d^{\, r_j, b}_1, d^{\, r_j, b}_2, d^{\, r_j, b}_3)_{j \in [1,\alpha], b \in \{0,1\}}$ denote the responses of the adversary $\adv$ 
  to these challenges, one can build a knowledge extractor $\ext$ as follows:

  \vspace{0.25\baselineskip}
  \pseudocode{%
    \tx{1. Let } \bar{\theta}^j = d^{\, r_j,0}_1, \tx{ compute } \bar{d}^{\, r_j,0}_1 \samples{\bar{\theta}^j} \hperm \tx{ for all } j \in [1, \alpha]\\
    \tx{2. Let } \bar{\xi}^j = d^{\, r_j,1}_1, \tx{ compute } \bar{d}^{\, r_j,1}_1 \samples{\bar{\xi}^j} \Ftn \tx{ for all } j \in [1, \alpha] \\
    \tx{3. Compute } \bm{c}_1 = \bar{d}^{\, r_1,1}_1 = \cdots = \bar{d}^{\, r_{\alpha},1}_1 \\
    \tx{4. Compute } \pi = \bar{d}^{\, r_1,0}_1 = \cdots = \bar{d}^{\, r_{\alpha},0}_1 \\ 
    \tx{5. Compute } \bm{c}_2 = \bm{H} d^{\, r_1,0}_2 - \rot{\bm{y}}{r_1} = \cdots = \bm{H} d^{\, r_{\alpha},0}_2 - \rot{\bm{y}}{r_{\alpha}} \\
    \tx{6. Compute } \bm{c}_3 = \bm{H}(\pi^{-1}[\bm{c}_1]) - \bm{c}_2 \\
    \tx{7. Compute } \bm{z}_j = \pi^{-1}[d^{\, r_j,1}_2] \tx{ for all } j \in [1, \alpha] \\
    \tx{8. Output } (\bm{c}_3, \bm{z}_1, \cdots, \bm{z}_{\alpha})
  }
  \vspace{0.5\baselineskip}

  \noindent We now prove that $(\bm{c}_3, \bm{z}_1, \cdots, \bm{z}_{\alpha})$ is a solution to the $\DSD(n, k, w, \alpha)$ problem. 
  If $\COM$ is a binding commitment scheme, one can compute $\bm{c}_1$ as all the $\bar{d}^{\, r_j,1}_1$ values are equals due to commitment $c_2$.
  Likewise, one can compute $\pi$ and $\bm{c}_2$ as all the values $\bar{d}^{\, r_j,0}_1$ and $\bm{H} d^{\, r_j,0}_2 - \rot{\bm{y}}{r_j}$ are equals due to commitment $c_1$.
  Using the same argument for commitment $c_3$, one has $\pi[d^{\, r_j,0}_2] = \bm{c}_1 + d^{\, r_j,1}_2$ hence $d^{\, r_j,0}_2 = \pi^{-1}[\bm{c}_1] + \pi^{-1}[d^{\, r_j,1}_2]$ for all $j \in [1,\alpha]$.
  From commitment $c_1$, one deduces that $\bm{c}_2 = \bm{H}(\pi^{-1}[\bm{c}_1]) + \bm{H}(\pi^{-1}[d^{\, r_j,1}_2]) - \rot{\bm{y}}{r_j}$ for all $j \in [1,\alpha]$. 
  Using quasi-cyclic codes, one gets $\rot{\bm{y}}{r_j} = \bm{H}\,\rot{\bm{x}}{r_j}$ which gives $\bm{H}\,\rot{\bm{x}}{r_j} = \bm{c}_3 + \bm{H} \bm{z}_j$ for all $j \in [1, \alpha]$.
  In addition, one have $\hw{\bm{z}_j} = w$ as $\hw{d^{\, r_j,1}_2} = w$  for all $j \in [1, \alpha]$.
  Therefore $(\bm{c}_3, \bm{z}_1, \cdots, \bm{z}_{\alpha})$ is a solution to the $\DSD(n, k, w, \alpha)$ problem.
  Using Theorem \ref{thm:reduction} completes the proof as there exists a reduction from $\QCSD(n, k, w)$ to $\DSD(n, k, w, \alpha)$ for an appropriate choice of~$\alpha$.

\end{proof}

\fi

%% file: pok_hvzk.tex
\vspace{0.5\baselineskip}
\begin{theorem}[\textbf{Honest Verifier Zero-Knowledge}] \label{thm:hvzk}
  If $\COM$ is a hiding commitment scheme, then the PoK depicted in Figure~\ref{fig:sd-structured2} satisfies the Honest-Verifier Zero-Knowledge property.
\end{theorem}
\vspace{0.5\baselineskip}

\iflong
\begin{proof} 
   We build a $\ppt$ simulator $\simulator$ that generates the view of an honest verifier with access to the public key $\pk = (\phi_2, \bm{y}^{1}, \cdots, \bm{y}^{s})$ only:\\
   \pseudocode{%
     \tx{1. } \simulator \tx{ samples } G_1 = ((s_1, r_1), \cdots, (s_\delta, r_\delta)) \sampler ([0, s - 1] \times [0, k-1])^\delta \tx{ and computes} \\ 
     \pcind \pcfor i \in \intoneto{\delta} \pcdo \\
     \pcind \pcind \theta_i \sampler \bit^{|\seed|}, ~ \pi_i \samples{\theta_i} \hperm, ~ \xi_i \sampler \bit^{|\seed|}, ~ \bm{v}_i \samples{\xi_i} \Ft^{2k}, ~ \bm{u}_i = \pi_i^{-1}[\bm{v}_i] \\
     \pcind \pcind \bm{H} \samples{\phi_2} \mathcal{QC}(\Ft^{\kttk}), ~ c_{i,1} = \cmt{\pi_i \, || \, \bm{H} \bm{u}_i}, ~ c_{i,2} = \cmt{\pi_i[\bm{u}_i]} \\
     \tx{2. } \simulator \tx{ sends } \Cmt_1 = \cmt{c_{1,1} || c_{1,2} || \cdots || c_{\delta,1} || c_{\delta,2}} \tx{ to } \adv \\
     \tx{3. } \adv \tx{ answers } \Ch_1 = ((s_1, r_1), \cdots, (s_\delta, r_\delta)) \tx{ to } \simulator \\
     \tx{4. } \tx{If } G_1 \neq \Ch_1 \tx{, rewind to step 1} \\
     \tx{5. } \simulator  \tx{ samples } G_2 = (b_1, \cdots, b_\delta)^\delta \sampler \bit^\delta \tx{ and computes} \\
     \pcind \pcfor i \in \intoneto{\delta} \pcdo \\
     \pcind \pcind \pcif b_i = 0 \pcthen \tx{ compute } \bm{\tilde{x}}^{s_i}_{r_i} \tx{ such that } \bm{H} \bm{\tilde{x}}^{s_i}_{r_i} = \rot{\bm{y}^{s_i}}{r_i} \tx{ (no weight constraint) } \\
     \pcind \pcind \pcif b_i = 1 \pcthen \tx{ compute } \bm{\tilde{x}}^{s_i}_{r_i} \sampler \Ftn \tx{ such that } \hw{\bm{\tilde{x}}^{s_i}_{r_i}} = w \\
     \pcind \pcind c_{i,3} = \cmt{\pi_i[\bm{u}_i  + \bm{\tilde{x}}^{s_i}_{r_i}]} \\
     \tx{6. } \simulator \tx{ sends } \Cmt_2 = \cmt{c_{1,3} || \cdots || c_{\delta,3}} \tx{ to } \adv. \\
     \tx{7. } \adv \tx{ answers } \Ch_2 = (b_1, \cdots, b_\delta) \tx{ to } \simulator. \\
     \tx{8. } \tx{If } G_2 \neq \Ch_2 \tx{, rewind to step 5} \\ 
     \tx{9. } \simulator \tx{ computes } \\
     \pcind \pcfor i \in \intoneto{\delta} \pcdo \\
     \pcind \pcind \pcif b_i = 0 \pcthen~ d_i = \big( \theta_i, \, \bm{u}_i + \bm{\tilde{x}}^{s_i}_{r_i}, \, c_{i,2} \big) \\
     \pcind \pcind \pcif b_i = 1 \pcthen~ d_i = \big( \xi_i, \, \pi_i[\bm{\tilde{x}}^{s_i}_{r_i}], \, c_{i,1} \big) \\
     \tx{10. } \simulator \tx{ sends } \Rsp = \compress{d_1, \cdots, d_\delta} \tx{ to } \adv. 
   }

  \noindent The view generated by $\simulator$ is $(\Cmt_1, \Ch_1, \Cmt_2, \Ch_2, \Rsp)$.
  If $\COM$ is an hiding commitment scheme, then $\Cmt_1$ and $\Cmt_2$ are indistinguishable in the simulation and during the real execution.
  In the honest verifier setting, $\Ch_1$ and $\Ch_2$ are generated similarly in the simulation and in the real execution.
  Thus, one only needs to check $\Rsp$.
  The simulator $\simulator$ behaves exactly as an honest prover $\prover$ except with respect to the values $(\bm{u}_i + \bm{\tilde{x}}^{s_i}_{r_i})$ and $\pi_i[\bm{\tilde{x}}^{s_i}_{r_i}]$.
  One can see that $(\bm{u}_i + \bm{x}^{s_i}_{r_i})$ and $(\bm{u}_i + \bm{\tilde{x}}^{s_i}_{r_i})$ follow the same probability distribution as $\bm{u}_i$ is sampled uniformly at random from $\Ftn$.
  In addition, $\pi_i[\bm{x}^{s_i}_{r_i}]$ and $\pi_i[\bm{\tilde{x}}^{s_i}_{r_i}]$ follow the same probability distribution as (i) $\bm{\tilde{x}}^{s_i}_{r_i}$ and $\bm{x}^{s_i}$ are sampled uniformly at random from the same distribution and (ii) $\bm{x}^{s_i}$ and $\bm{x}^{s_i}_{r_i}$ follow the same probability distribution.
  Finally, one can see that the simulator runs in polynomial time with $2 \times s \times k \times \delta$ expected rewinds.
\end{proof}

\fi